\documentclass[11pt]{article}

\usepackage[a4paper,margin=2.5cm]{geometry}
\usepackage{amsmath,amssymb,amsthm}
\usepackage{tikz}
\usepackage{microtype}

\newtheorem{proposition}{Proposition}

\title{Local Complex Dependence and Separability in Madelung Hydrodynamics}
\author{Lorenzo Pirovano \\ \emph{lorenzopirovano.res@gmail.com}}
\date{}

\begin{document}
	\maketitle
	
	\begin{abstract}
		For a many-particle pure state in a fixed position representation, we consider
		the mixed cross-particle derivatives of the logarithm of the wave function.
		Their real part is one half of the Holland--Wang local dependence function of
		the configuration density, while their imaginary part is the cross-Jacobian of
		the Madelung velocity field.  On a node-free product region, vanishing of all
		cross blocks throughout the region is equivalent to local multiplicative
		separability.  Under Schr\"odinger evolution with a real scalar potential, the
		initial growth of a cross block from a separable state is sourced by the
		corresponding mixed Hessian of the potential and is purely imaginary to first
		order in time.  The construction is a local separability and dependence
		diagnostic, rather than a basis-independent entanglement measure.
	\end{abstract}
	
	\section{Bipartite Madelung state}
	
	Fix a time and consider two distinguishable particles in one spatial dimension.
	Using the Madelung representation~\cite{Madelung1927},
	\begin{equation}
		\psi(x,y)=\sqrt{\rho(x,y)}\,e^{iS(x,y)/\hbar}.
	\end{equation}
	The phase is understood locally on regions where $\psi\neq0$.  The Madelung
	velocity components are
	\begin{equation}
		v_x=\frac{1}{m_x}\frac{\partial S}{\partial x},
		\qquad
		v_y=\frac{1}{m_y}\frac{\partial S}{\partial y}.
	\end{equation}
	A product state has the form
	\begin{equation}
		\psi(x,y)=\psi_x(x)\psi_y(y).
	\end{equation}
	On a node-free product region, multiplicative separability is equivalent to
	additive separability of a consistent local branch of $\log\psi$.
	
	The purpose of this note is not to define a new global entanglement measure,
	but to combine three structures that are usually treated separately: the
	Holland--Wang mixed logarithmic derivative of the configuration density, the
	mixed phase Hessian controlling cross-particle Madelung velocities, and the
	rank-one condition for multiplicative separability.  Packaging them as the
	complex off-diagonal Hessian of $\log\psi$ gives a pointwise diagnostic that
	extends directly to multipartite cluster separability.
	
	\section{Local complex dependence}
	
	Define, wherever $\psi\neq0$,
	\begin{equation}
		\mathcal K_{xy}(x,y)
		=\frac{\partial^2}{\partial x\,\partial y}\log\psi(x,y).
	\end{equation}
	Equivalently,
	\begin{equation}
		\mathcal K_{xy}
		=\frac{\psi\,\partial_x\partial_y\psi
			-(\partial_x\psi)(\partial_y\psi)}{\psi^2},
	\end{equation}
	which makes explicit that the logarithmic derivative is independent of the
	choice of local logarithm branch.  The Madelung decomposition gives
	\begin{equation}
		\mathcal K_{xy}
		=\frac12\frac{\partial^2\log\rho}{\partial x\,\partial y}
		+\frac{i}{\hbar}\frac{\partial^2S}{\partial x\,\partial y}.
	\end{equation}
	Hence
	\begin{equation}
		\operatorname{Re}\mathcal K_{xy}
		=\frac12 H_{\rm HW}(x,y),
		\qquad
		H_{\rm HW}=\frac{\partial^2\log\rho}{\partial x\,\partial y},
	\end{equation}
	where $H_{\rm HW}$ is the local dependence function introduced in statistics by Holland and	Wang~\cite{HollandWang1987} and further studied by Jones~\cite{Jones1996}.
	The imaginary part has the direct hydrodynamic interpretation
	\begin{equation}
		\frac{\partial v_x}{\partial y}
		=\frac{\hbar}{m_x}\operatorname{Im}\mathcal K_{xy},
		\qquad
		\frac{\partial v_y}{\partial x}
		=\frac{\hbar}{m_y}\operatorname{Im}\mathcal K_{xy}.
	\end{equation}
	It therefore gives the pointwise cross-response of one Madelung velocity sector
	to the coordinate of the other.
	
	The same field can be written in terms of the complex quantum momentum
	\begin{equation}
		p_x^{\rm c}=-i\hbar\,\partial_x\log\psi,
	\end{equation}
	as
	\begin{equation}
		\partial_y p_x^{\rm c}=-i\hbar\,\mathcal K_{xy}.
	\end{equation}
	Complex logarithmic momentum and velocity fields are standard objects in
	complex quantum hydrodynamics~\cite{BonillaSchuch2021}.  The point here is the
	cross-particle block and its simultaneous statistical and hydrodynamic
	interpretation.
	
	Related Bohmian work has introduced hierarchies of ``entanglement fields''
	built from conditional wave functions and derivatives of the configuration-space
	wave function in attempts to formulate pilot-wave dynamics in physical
	space~\cite{Norsen2010,NorsenMarianOriols2015}.  These are distinct from
	$\mathcal K_{xy}$: the latter is defined directly on configuration space as a
	mixed logarithmic derivative and is used here as a local dependence and
	separability diagnostic, rather than as a hierarchy of conditional dynamical
	fields.
	
	\begin{proposition}
		Let $U$ and $V$ be open, connected intervals and let
		$\psi\in C^2(U\times V)$ be nowhere zero.  Then
		\begin{equation}
			\mathcal K_{xy}=0 \quad\text{throughout }U\times V
		\end{equation}
		if and only if $\psi(x,y)=a(x)b(y)$ on $U\times V$.
	\end{proposition}
	
	\begin{proof}
		Since $U$ and $V$ are intervals, $U\times V$ is simply connected.  Because
		$\psi$ is nowhere zero there, it admits a consistent $C^2$ logarithm $L$ on
		$U\times V$.  The forward implication then follows from
		$\partial_y(\partial_x L)=0$: hence $\partial_xL$ depends only on $x$, and
		integration gives $L=f(x)+g(y)+c$.  Exponentiation gives the product form.
		The converse is immediate by differentiation.
	\end{proof}
	
	Thus a zero of $\mathcal K_{xy}$ at a single point is not by itself a
	separability statement: separability follows when the field vanishes
	throughout a node-free product region.
	
	\section{Finite separability ratio}
	
	The infinitesimal field has a simple finite counterpart.  Define, wherever the
	denominator is nonzero,
	\begin{equation}
		\mathcal R_\psi(x,x';y,y')=
		\frac{\psi(x,y)\psi(x',y')}
		{\psi(x,y')\psi(x',y)}.
	\end{equation}
	All one-particle multiplicative factors cancel.  Moreover,
	\begin{equation}
		\mathcal R_\psi=1
		\quad\Longleftrightarrow\quad
		\psi(x,y)\psi(x',y')-\psi(x,y')\psi(x',y)=0,
	\end{equation}
	so the finite condition is the usual rank-one condition written as a
	cross-ratio.  On a node-free product domain, $\mathcal R_\psi=1$ for all
	quadruples if and only if the state is multiplicatively separable.
	
	If the rectangle joining $(x,y)$ to $(x',y')$ lies in a node-free patch and a
	consistent branch of $\log\psi$ is chosen there, then
	\begin{equation}
		\log\mathcal R_\psi
		=\int_x^{x'}\!du\int_y^{y'}\!dv\;
		\frac{\partial^2\log\psi(u,v)}{\partial u\,\partial v}.
	\end{equation}
	Here the logarithm on the left denotes the value induced by the same chosen
	branch of $\log\psi$; without this convention, the identity is understood
	modulo $2\pi i$.
	Consequently,
	\begin{equation}
		\mathcal K_{xy}
		=\lim_{\delta x,\delta y\to0}
		\frac{\log\mathcal R_\psi(x,x+\delta x;y,y+\delta y)}
		{\delta x\,\delta y}.
	\end{equation}
	The modulus of $\mathcal R_\psi$ probes density dependence, while its argument
	probes the corresponding phase mismatch.
	
	\section{Minimal phase-coupled example}
	
	Consider
	\begin{equation}
		\psi(x,y)=\phi(x)\chi(y)e^{i\lambda xy}.
	\end{equation}
	Here $\lambda$ has the reciprocal units of $xy$, so that the phase is
	dimensionless.
	Its position density factorizes exactly,
	\begin{equation}
		|\psi(x,y)|^2=|\phi(x)|^2|\chi(y)|^2,
	\end{equation}
	but
	\begin{equation}
		\mathcal K_{xy}=i\lambda.
	\end{equation}
	Thus the instantaneous position density has no cross-dependence, whereas the
	velocity field is coupled:
	\begin{equation}
		\frac{\partial v_x}{\partial y}=\frac{\hbar\lambda}{m_x},
		\qquad
		\frac{\partial v_y}{\partial x}=\frac{\hbar\lambda}{m_y}.
	\end{equation}
	For finite-width factors the state is nonseparable when $\lambda\neq0$.  If
	$|\phi|^2$ and $|\chi|^2$ are centered Gaussians with variances
	$\sigma_x^2$ and $\sigma_y^2$, respectively, then
	\begin{equation}
		\operatorname{Tr}(\widehat\rho_x^2)
		=\frac{1}{\sqrt{1+4\lambda^2\sigma_x^2\sigma_y^2}}.
	\end{equation}
	This also shows directly that $\mathcal K$ is not a global entanglement
	measure: at fixed $\lambda$ the local field remains $i\lambda$, while changing
	the widths changes the reduced-state purity.
	
	A closely related Gaussian family displays the two parts of $\mathcal K$ at
	once.  Let
	\begin{equation}
		\psi_{\kappa,\lambda}(x,y)
		=\mathcal N\exp\!\left[
		-\frac{x^2}{4\sigma_x^2}
		-\frac{y^2}{4\sigma_y^2}
		+(\kappa+i\lambda)xy\right],
	\end{equation}
	where $\mathcal N$ normalizes the state and
	$|\kappa|<1/(2\sigma_x\sigma_y)$.  Then
	\begin{equation}
		\mathcal K_{xy}=\kappa+i\lambda.
	\end{equation}
	Thus $\kappa$ produces a uniform Holland--Wang density dependence, whereas
	$\lambda$ produces a uniform cross-response of the Madelung velocities.  Both
	parameters have the reciprocal units of $xy$.
	
	\section{Relation to configuration and phase entanglement}
	
	Zander and Plastino~\cite{ZanderPlastino2018} decompose the linear entropy of a
	pure bipartite state into configuration and phase indicators,
	\begin{equation}
		E_{\rm lin}=1-\operatorname{Tr}(\widehat\rho_x^2)=E_c+E_p.
	\end{equation}
	The former is associated with nonfactorizability of the configuration density,
	and the latter with nonadditivity of the phase.  These are global scalar
	quantities.  By contrast,
	\begin{equation}
		\operatorname{Re}\mathcal K_{xy}
		=\frac12\partial_x\partial_y\log\rho,
		\qquad
		\operatorname{Im}\mathcal K_{xy}
		=\frac1\hbar\partial_x\partial_yS
	\end{equation}
	retains the location and differential structure of the coupling.  In the
	phase-coupled example above the density contribution vanishes, while
	$\mathcal K_{xy}=i\lambda$ identifies a uniform phase coupling.
	
	There is no universal reconstruction of global entanglement from $\mathcal K$
	alone.  Indeed, for nonvanishing one-particle factors $a(x)$ and $b(y)$ chosen
	so that the filtered state remains normalizable,
	\begin{equation}
		\psi'(x,y)=C\,a(x)b(y)\psi(x,y)
	\end{equation}
	satisfies $\mathcal K'_{xy}=\mathcal K_{xy}$.  Pure-phase factors are local
	unitaries, but position-dependent moduli implement local nonunitary filtering
	and can change the Schmidt coefficients after normalization.  Thus the same
	local field may correspond to different global entanglement.
	
	\section{Many particles in three dimensions}
	
	For $N$ distinguishable spinless particles,
	\begin{equation}
		\psi(\mathbf x_1,\ldots,\mathbf x_N)
		=\sqrt{\rho}\,e^{iS/\hbar},
	\end{equation}
	define, for $i\neq j$,
	\begin{equation}
		\mathcal K_{ij}^{ab}
		=\frac{\partial^2\log\psi}
		{\partial x_i^a\,\partial x_j^b},
		\qquad a,b\in\{1,2,3\}.
	\end{equation}
	Each $\mathcal K_{ij}$ is a complex $3\times3$ cross block, with
	\begin{equation}
		\frac{\partial v_i^a}{\partial x_j^b}
		=\frac{\hbar}{m_i}\operatorname{Im}\mathcal K_{ij}^{ab}.
	\end{equation}
	The collection of these matrices forms the off-diagonal block structure of
	the complex Hessian of $\log\psi$ and therefore organizes the local many-body
	coupling by particle sectors.  Equivalently, if
	$p_{i,a}^{\rm c}=-i\hbar\,\partial_{x_i^a}\log\psi$, then
	\begin{equation}
		\partial_{x_j^b}p_{i,a}^{\rm c}
		=-i\hbar\,\mathcal K_{ij}^{ab}.
	\end{equation}
	The cross blocks also identify separable clusters, not only complete
	factorization into one-particle states.
	
	\begin{proposition}[Cluster separability]
		\label{prop:cluster-separability}
		Let $\Omega=U_1\times\cdots\times U_N$ be an open, simply connected, node-free
		product region, and let $A$ and $B$ be a nontrivial partition of
		$\{1,\ldots,N\}$.  Then
		\begin{equation}
			\mathcal K_{ij}=0
			\qquad\text{throughout }\Omega,
			\quad i\in A,\ j\in B,
		\end{equation}
		if and only if
		\begin{equation}
			\psi(\mathbf X_A,\mathbf X_B)
			=\psi_A(\mathbf X_A)\psi_B(\mathbf X_B)
		\end{equation}
		on $\Omega$.
	\end{proposition}
	
	\begin{proof}
		Since $\psi$ is nonzero on the simply connected region, a consistent logarithm
		may be chosen.  Vanishing of all mixed derivatives across the $A|B$ partition
		implies that the $\mathbf X_A$-gradient of $\log\psi$ is independent of
		$\mathbf X_B$, and conversely.  Hence
		$\log\psi=f_A(\mathbf X_A)+f_B(\mathbf X_B)+c$, which gives the stated product
		form.  The reverse implication follows immediately by differentiation.
	\end{proof}
	
	In particular, define a regional coupling graph $G_\Omega$ whose vertices are
	the particle labels and in which $i$ and $j$ are joined when
	$\mathcal K_{ij}$ does not vanish identically on $\Omega$.  If its connected
	components are $C_1,\ldots,C_r$, the proposition gives
	\begin{equation}
		\psi=\prod_{\alpha=1}^r \psi_{C_\alpha}(\mathbf X_{C_\alpha})
		\qquad\text{on }\Omega.
	\end{equation}
	
	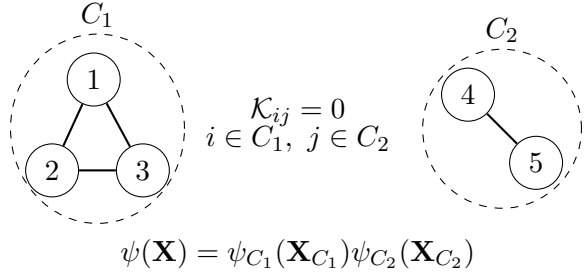
\begin{figure}[t]
		\centering
		\begin{tikzpicture}[
			particle/.style={circle,draw,minimum size=7mm,inner sep=0pt},
			coupling/.style={thick}
			]
			\node[particle] (p1) at (-2.7,0.65) {$1$};
			\node[particle] (p2) at (-3.25,-0.55) {$2$};
			\node[particle] (p3) at (-2.05,-0.55) {$3$};
			\draw[coupling] (p1)--(p2);
			\draw[coupling] (p1)--(p3);
			\draw[coupling] (p2)--(p3);
			\draw[dashed] (-2.65,0) ellipse [x radius=1.15,y radius=1.25];
			\node at (-2.65,1.5) {$C_1$};
			
			\node[particle] (p4) at (2.25,0.45) {$4$};
			\node[particle] (p5) at (3.15,-0.45) {$5$};
			\draw[coupling] (p4)--(p5);
			\draw[dashed] (2.7,0) ellipse [x radius=1.05,y radius=1.05];
			\node at (2.7,1.3) {$C_2$};
			
			\node[align=center] at (0,0.05)
			{$\mathcal K_{ij}=0$\\[-1mm]
				$i\in C_1,\ j\in C_2$};
			\node at (0,-1.65)
			{$\psi(\mathbf X)=\psi_{C_1}(\mathbf X_{C_1})
				\psi_{C_2}(\mathbf X_{C_2})$};
		\end{tikzpicture}
		\caption{Schematic regional coupling graph.  Solid edges represent
			cross-particle blocks that may be nonzero within a connected component.  The
			absence of cross edges between $C_1$ and $C_2$ means
			$\mathcal K_{ij}=0$ across the cut throughout $\Omega$, implying cluster
			separability on that region.}
		\label{fig:coupling-graph}
	\end{figure}
	
	Complete one-particle separability is the limiting case in which
	$G_\Omega$ has no edges.  Thus the cross-block structure gives a local
	criterion for both full separability and exact cluster decomposition.
	
	\section{Dynamical generation of dependence and entanglement}
	
	The cross blocks also give a direct description of how nonseparability is
	generated dynamically.  Let the wave function evolve, on a node-free region,
	according to
	\begin{equation}
		i\hbar\,\partial_t\psi
		=
		\left[
		-\sum_{k=1}^N\frac{\hbar^2}{2m_k}\nabla_k^2
		+V(\mathbf X,t)
		\right]\psi,
	\end{equation}
	and write $L=\log\psi$.  Dividing the Schr\"odinger equation by $\psi$ gives
	\begin{equation}
		i\hbar\,\partial_t L
		=
		-\sum_{k=1}^N\frac{\hbar^2}{2m_k}
		\left[
		\nabla_k^2L+\nabla_kL\cdot\nabla_kL
		\right]
		+V.
	\end{equation}
	For a sufficiently smooth solution, taking mixed derivatives yields the exact
	local evolution equation
	\begin{align}
		i\hbar\,\partial_t\mathcal K_{ij}^{ab}
		={}&
		-\sum_{k=1}^N\frac{\hbar^2}{2m_k}
		\partial_{x_i^a}\partial_{x_j^b}
		\left[
		\nabla_k^2L+\nabla_kL\cdot\nabla_kL
		\right]
		\nonumber\\
		&+
		\partial_{x_i^a}\partial_{x_j^b}V,
		\qquad i\neq j.
	\end{align}
	The mixed Hessian of the potential therefore appears as a direct source for
	cross-particle dependence, while the remaining terms describe its subsequent
	kinetic evolution.
	
	\begin{proposition}[Initial generation across a separable cut]
		\label{prop:initial-generation}
		Suppose that at time $t_0$ the state is multiplicatively separable across a
		partition $A|B$ on a node-free product region,
		\begin{equation}
			\psi(\mathbf X_A,\mathbf X_B,t_0)
			=\psi_A(\mathbf X_A,t_0)\psi_B(\mathbf X_B,t_0).
		\end{equation}
		Then, for $i\in A$ and $j\in B$,
		\begin{equation}
			\left.
			\partial_t\mathcal K_{ij}^{ab}
			\right|_{t_0}
			=
			-\frac{i}{\hbar}
			\left.
			\frac{\partial^2V}
			{\partial x_i^a\,\partial x_j^b}
			\right|_{t_0}.
		\end{equation}
	\end{proposition}
	
	\begin{proof}
		At $t_0$, separability gives
		$L=L_A(\mathbf X_A)+L_B(\mathbf X_B)$.  For each $k$, the quantity
		$\nabla_k^2L+\nabla_kL\cdot\nabla_kL$ therefore depends only on the
		coordinates of the cluster containing $k$.  Its mixed derivative across the
		$A|B$ cut vanishes, leaving only the potential term in the evolution equation
		above.
	\end{proof}
	
	For a real scalar potential, the leading departure from separability is thus
	purely imaginary:
	\begin{equation}
		\mathcal K_{ij}^{ab}(t_0+\delta t)
		=
		-\frac{i\,\delta t}{\hbar}
		\left.
		\frac{\partial^2V}
		{\partial x_i^a\,\partial x_j^b}
		\right|_{t_0}
		+O(\delta t^2).
	\end{equation}
	Hence $\operatorname{Re}\mathcal K_{ij}=O(\delta t^2)$, whereas the phase
	sector generally appears already at first order.  If the mixed Hessian of
	$V$ is nonzero across the cut, the state cannot remain multiplicatively
	separable there for sufficiently small nonzero times.  Conversely, when the
	Hamiltonian itself separates across the cut, all such source terms vanish and
	an initially product state remains a product state.
	For an initially globally separable pure state, loss of multiplicative
	separability on any such region is therefore the onset of entanglement.
	
	As a minimal example, consider two particles in one dimension and an ideal
	interaction pulse generated by $V(x,y)=gxy$, with free evolution neglected
	during the pulse.  Starting from $\psi(x,y,0)=\phi(x)\chi(y)$ gives exactly
	\begin{equation}
		\psi(x,y,t)
		=
		\phi(x)\chi(y)e^{-igtxy/\hbar},
		\qquad
		\mathcal K_{xy}(t)=-\frac{igt}{\hbar}.
	\end{equation}
	Thus the phase-coupled example above is generated dynamically with
	$\lambda=-gt/\hbar$: the position density remains factorized during the ideal
	pulse, while
	\begin{equation}
		\frac{\partial v_x}{\partial y}=-\frac{gt}{m_x},
		\qquad
		\frac{\partial v_y}{\partial x}=-\frac{gt}{m_y}.
	\end{equation}
	For finite-width factors and $gt\neq0$, the resulting pure state is entangled:
	the interaction has generated nonseparability entirely through the phase while
	leaving the instantaneous position density factorized.
	
	\section{Scope}
	
	The field $\mathcal K_{ij}$ is a local diagnostic in a chosen position
	representation, not a basis-independent entanglement monotone.  Pure-state
	separability itself is invariant under local changes of subsystem basis, but
	the pointwise field $\mathcal K_{ij}$, its numerical magnitude, and its regional
	structure are representation dependent.  The zero cross-block condition on a
	node-free product region therefore provides a representation-specific local
	criterion for the underlying product structure.  Nodes require a patchwise
	treatment and may carry
	additional topological information.  Identical particles, spin, mixed states,
	and a systematic analysis of long-time cross-block dynamics require separate
	treatment.  The coupling graph suggests adaptive cluster descriptions
	when the cross-block structure is sparse, but no dimensional reduction is
	claimed here: in particular, small nonzero blocks do not by themselves provide
	a controlled approximation error.

\end{document}